\documentclass[10pt,a4paper]{article}
\usepackage{hyperref}
\usepackage{amsmath,graphicx}
\usepackage{tabularx}
\usepackage{psfrag}
\usepackage{xcolor}
\usepackage{amsfonts}
\usepackage{amssymb}
\usepackage{caption}
\usepackage{subcaption}
\usepackage{booktabs}
\usepackage{multirow}
\usepackage{enumerate}
\usepackage{cite}		
\usepackage{amsthm}
\DeclareGraphicsExtensions{.jpg,.mps,.pdf,.png,.eps} 
\usepackage{array}
\usepackage{nameref,hyperref}
\usepackage[numbers, comma, sort&compress]{natbib}
\usepackage{xspace}		
\usepackage{algorithm} 
\usepackage{algorithmic}

\usepackage{url}

\include{macros}
\renewcommand{\nu}{\vec{\mathbf{n}}}

\newcommand{\RR}{\mathbb{R}}

\newcommand{\PP}{\mathbb{P}}

\newtheorem{proposition}{Proposition}

\newtheorem{corollary}{Corollary}

\newtheorem{lemma}{Lemma}
\newtheorem{theorem}{Theorem}

\title{NonSmooth Convex Optimization to Estimate the Covid19 Reproduction Number Space-Time Evolution with Robustness against Outliers}

\author{Barbara Pascal\footnotemark[1] \and Patrice Abry\footnotemark[2] \and
Nelly Pustelnik\footnotemark[2] \and St\'ephane Roux\footnotemark[2] \and R\'emi Gribonval\footnotemark[3] \and
Patrick Flandrin\footnotemark[2]}

\renewcommand{\thefootnote}{\fnsymbol{footnote}}
\begin{document}

\maketitle

\footnotetext[1]{B. Pascal is with Univ. Lille, CNRS, Centrale Lille, UMR 9189 CRIStAL, F-59000 Lille, France (e-mail: barbara.pascal@univ-lille.fr).}
\footnotetext[2]{P. Abry, N. Pustelnik, S. Roux, P. Flandrin are with CNRS, ENS de Lyon, Laboratoire de Physique, Lyon, France  (e-mail: firstname.lastname@ens-lyon.fr).}
\footnotetext[3]{R. Gribonval is with INRIA, ENS de Lyon, Laboratoire d'Informatique.}

\begin{abstract}
Daily pandemic surveillance, often achieved through the estimation of the reproduction number, constitutes a critical challenge for national health authorities to design counter-measures. 
In an earlier work, we proposed to formulate the estimation of the reproduction number as an optimization problem, combining data-model fidelity and space-time regularity constraints, solved by nonsmooth convex proximal minimizations. 
Though promising, that first formulation significantly lacks robustness against the Covid-19 data low quality (irrelevant or missing counts, pseudo-seasonalities,\ldots) stemming from the emergency and crisis context, which significantly impairs accurate pandemic evolution assessments. 
The present work aims to overcome these limitations by carefully crafting a functional permitting to estimate jointly, in a single step, the reproduction number and outliers defined to model low quality data.
This functional also enforces epidemiology-driven regularity properties for the reproduction number estimates, while preserving convexity, thus permitting the design of efficient minimization algorithms, based on proximity operators that are derived analytically. 
The explicit convergence of the proposed algorithm is proven theoretically. 
Its relevance is quantified on real Covid-19 data, consisting of daily new infection counts for 200+ countries and for the 96 metropolitan France counties, publicly available at Johns Hopkins University and Sant\'e-Publique-France. 
The procedure permits automated daily updates of these estimates, reported via animated and interactive maps. 
Open-source estimation procedures will be made publicly available..
\end{abstract}

\renewcommand{\thefootnote}{\arabic{footnote}}

\section{Introduction}
\noindent {\bf Context.} 
The ongoing COVID-19 pandemic has produced an unprecedented health and economic crisis, urging for the construction of efficient monitoring procedures of its spreading across territories, a crucial step in designing  sanitary, social and economic policies by national authorities \cite{flahault2020covid}. 
It is however often not the value of the infection level per se that matters to design pandemic counter-measures (lockdown,\ldots), but its evolution along time and variations across territories. 
However, in the context of the Covid-19 pandemic outburst and of the sanitary crisis, collecting daily new infection counts, the basis material in any pandemic surveillance strategy, results in low-quality data (missing samples, outliers, seasonalities\ldots).
Surprisingly, after 18 months of pandemic, the quality of the data collected and made available by most national health authorities in the world remains limited, which severely impairs an accurate and timely pandemic evolution assessment, the issue at the heart of the present work. \\
\noindent {\bf Related works.} Pandemic surveillance has been envisaged with several categories of tools from different fields of sciences (cf. \cite{arino2021describing} for a review). 
Yet, it is classically performed using {\it compartmental models}, elaborating on the classical \emph{Susceptible-Infectious-Recovered} reference.
For realistic use and to match social realities (social groups, contact, \ldots), these models need to be refined (cf. e.g., \cite{Liu2018,BCF}), essentially by increasing the number of compartments, which  implies increasing (in a quadratic way) the numbers of parameters to be estimated.
Parameter estimation is usually achieved within Bayesian frameworks, maximizing the likelihood attached to the models, 
at the expense of heavy computational burdens. 
Such models are thus often used a posteriori (i.e., after the epidemic) with consolidated and accurate datasets.
The low quality of the Covid-19 data collected by most national public health authorities across the world significantly impairs the use of such Bayesian schemes and thus the reliable estimations of these parameters.
This thus precludes the use of such models for daily basis update of the pandemic evolution assessment. 

Instead of compartmental models, epidemiologists massively use the so-called {\it reproduction number}, $R$, that measures how many new individuals are on average infected by an already contaminated person, a key marker of the pandemic strength (cf. e.g., \cite{Diekmann1990,wallinga2004,vandenDriessche2002,obadia2012,cori2013new}).
Estimated to around $3$ during the outburst of the Covid-19 pandemic \cite{DiDomenico2020,salje}, it is also used as a function of  time $R_T$ to monitor the temporal evolution of the pandemic. 
The strength of such a pandemic monitoring is to be based on a single parameter to estimate $R$, while accounting for the basic and universal pandemic propagation mechanisms:  
The number of new infections, $Z_T$, at day $T$, depends on $R_T$ and on the numbers of new infections measured at previous days, $\{Z_{T-1}, Z_{T-2}, Z_{T-3},\ldots\} $ weighted by the so-called  \emph{serial interval function} $\Phi(t)$.
The latter quantifies the probability that someone found Covid-positive today was actually contaminated several days ago \cite{cori2013new,obadia2012,thompson2019,Liu2018}.
Following \cite{cori2013new}, $R_T$ can be estimated using a Bayesian scheme\footnote{\href{https://shiny.dide.imperial.ac.uk/epiestim/}{https://shiny.dide.imperial.ac.uk/epiestim/}}. 
Instead, in an earlier work \cite{abry2020spatial}, we proposed an estimation procedure of the spatio-temporal evolution of the reproduction number, written as a nonsmooth convex optimization problem, combining data-model fidelity with time and space regularity constraints for the estimates of $R$ to be pandemic realistic, and solved by proximal algorithms \cite{Bauschke:2011ta,Cai_JF_2012_j-ams_ima_rtv,Combettes2011,parikh2014proximal,pustelnik2012_j-ieee-tsp_sur_ads}. 
Though promising, that proximal-optimization based estimation significantly lacks robustness against the low quality of the Covid-19 data, a generic and crucial issue addressed in the present contribution. \\
\noindent {\bf Goals, contributions and outline.} The overall goal of this work is to devise an inverse problem-type convex-optimization procedure to assess the spatio-temporal evolution of  the reproduction number $R$, from daily new infection counts, measured simultaneously on several connected territories, such as the different counties of a same country. 
The proposed estimation procedure is designed to be: i)  robust to the poor quality of Covid-19 data~; ii) efficiently applicable \emph{online} to update estimates, on a daily basis or whenever new data are available, at moderate computational costs even for large datasets. 

To that end, Section~\ref{sec:model} recalls the pandemic model used here \cite{cori2013new}, frames its classical estimation and related issues. 
Section~\ref{sec:estim} describes the core methodological contributions of the paper: 
i) It discusses the poor quality of the data and proposes an extension of the original pandemic modelling that handles poor quality data as \emph{outliers} (cf. Section~\ref{sec-outlier})~;
ii) It proposes an original and theoretical design of a functional whose minimization yields in a single step a joint estimation of both the reproduction number and the outliers, by combining the log-likelihood of the extended model as a data fidelity term with temporal and spatial regularity constraints for epidemiology-realistic estimates of $R$ and sparsity constraints for the outliers (cf. Section~\ref{ssec:multiple}).
This functional is carefully crafted to preserve convexity so that fast and efficient proximal algorithms can be devised for minimization~; 
iii) The theoretical properties of the functional and of its solution (existence, unicity) are studied theoretically in Section~\ref{ssec:convex}~;
iv) This permits to devise a proximal type iterative algorithm, with explicit analytical derivation of the proximal operators, whose convergence and stopping criterion are discussed theoretically and practically (cf. Section~\ref{ssec:min_prox}). 

The relevance  of the proposed estimation procedure to assess the spatio-temporal evolution of the pandemic is illustrated on real Covid-19 data (daily counts of new infections) made available by national public health authorities (such as e.g.,  \emph{Sant\'e-Publique-France}, SPF) or collected by Johns Hopkins University, described in Section~\ref{sec:data}.

Estimation performance are reported in Section~\ref{sec:results}:
i) First, the role and choices of the hyperparameters balancing the regularization terms and the data-model fidelity term are discussed in detail (cf. Section~\ref{sec:hyperparam})~; 
ii) Second, the benefits of the proposed strategy to estimate the temporal evolution of $R_T$ at the level of a given country are illustrated by comparisons against different estimation strategies, such as a two-step procedure (outliers estimated and removed first, followed by estimation of $R$)  (cf. Section~\ref{sec:country})~;
iii) Third,  the relevance of the space-time evolution of $R$ estimated across territories that are \emph{connected} (shared borders, massive population commutations,\ldots), here the 96 \emph{d\'epartements} (counties) that organize administratively metropolitan France), is illustrated and discussed precisely (cf. Section~\ref{sec:france}). 

The proposed procedure is efficient enough to permit the automated daily updates of the estimates of $R_T$ for 200+ countries, for  the 96 continental French \emph{d\'epartements} of France and for the 50 US states. 
Links to the resulting animated and interactive maps of estimates are provided and open-source estimation procedures will be made publicly available. 

\section{Pandemic Model and Estimation}
\label{sec:model}
\subsection{Model} 

The pandemic model used here, focused on the reproduction number $R$, was proposed in \cite{cori2013new} and relies on two key ideas: 
i) Conditionally to past counts $Z_{1:T-1} \triangleq  \{ Z_1,\ldots,Z_{T-1} \}$, the count of daily new infections, $Z_T$, is modeled as a random variable drawn from a Poisson distribution $\mathrm{Poiss}(p_T)$, 
ii) whose parameter $p_T$ depends on the past counts $Z_{1:T-1}$, on the causal \emph{serial interval} function $ \Phi_T$ and on the current $R_T$:
\begin{equation}
\label{equ:cori}
\PP \left(Z_T \vert Z_{1:T-1}; R_T\right)
\sim \mathrm{Poiss}\left( p_T \triangleq  R_T \times \sum_{t = 1}^{\tau_\phi} \Phi_t
Z_{T-t}  \right)
\end{equation}

\noindent The serial interval function $\Phi_t$ is a key element of the model that accounts for the epidemiology evolution mechanisms, as it models the random delays between the onset of symptoms in a primary case and the onset of symptoms in secondary cases, 
\cite{cori2013new,obadia2012,thompson2019,Liu2018}.
For the Covid-19 pandemic and for earlier pandemics of same types, it was shown that  $\Phi$ can be approximated as a Gamma function, with shape and rate parameters 1.87 and 0.28, respectively, corresponding to mean and standard deviations of 6.6 and 3.5 days, indicating a high risk of infecting other persons from 3 to 10 days after the symptoms have appeared  \cite{shujuan,Riccardo2020,Guzzetta}.  
It is assumed here that $\Phi$ is known and follows this Gamma approximation.

\subsection{Estimation} 

Given the  analytic expression of the Poisson distribution in Model~\eqref{equ:cori}, the log-likelihood $\ln \PP \left(Z_T \vert Z_{1:T-1}; R_{1:T}\right) =  Z_T \ln p_T - p_T - \ln Z_T~! $, permits to compute 
 the Maximum-Likelihood Estimate, 
\begin{equation}
\label{eq:CoriMLE}
\hat R^{\text{MLE}}_T = {\mathrm{argmin}} (- \ln \PP \left(Z_T \vert Z_{1:T-1}; R_{1:T}\right)),
\end{equation}
in closed-form expression as: $\hat R^{\text{MLE}}_T  =   {Z_T/ \sum_{t\geq 1} \Phi_t Z_{T - t}}$.
It can be interpreted as a ratio of moving averages, whose size and shape are consistent with pandemic modeling (thus with $\Phi_T)$.
The $\hat R^{\text{MLE}}_T $ computed from real Covid-19 data (cf. Section~\ref{sec:data}) are displayed in Figs.~\ref{fig:france} and ~\ref{fig:countries} (middle plots) and display a far too large variability along time to be of any practical use in actual pandemic monitoring.  
Such variability mostly stems from the low quality of the data.
The goal of the present work is thus to improve the estimation of $R$ despite the low quality of the available Covid-19 data.

\section{Robust estimation: Methodology}
\label{sec:estim}
 \subsection{Covid-19 data low quality and multivariate infection counts} 
 \label{sec-outlier}

Covid-19 daily new infection counts made available by public health authorities are for most countries severely corrupted, 
with missing samples, non meaningful negative counts followed by retrospective cumulated counts spread across the following days,\ldots.
They also show pseudo-seasonality effects, with significantly less counts on non working days. 
Seasonality are however uneasy to model, as non working days are not only week-ends but also additional day-offs that depend on countries.
Also, the way the pandemic has been monitored along week-ends and non-working days has significantly changed along time, even within a given country. 
Therefore, rather than trying to devise necessarily ad-hoc and parametric models for noise in Covid-19 data, we have instead chosen to refer to these different types of \emph{data corruption} under the generic term of \emph{outlier}, denoted  $O_T$, and we propose to model observed daily new infection counts $Z_T$ still from a Poisson distribution conditionally to past counts $Z_{1:T-1}$, yet with a Poisson parameter $p_T$ 
that depends both on the current reproduction number $R_{T}$ and current \textit{outliers} $O_{T}$. 

Additionally, when new infection counts are monitored \textit{simultaneously} for several connected territories, e.g., the counties or states 
of a given country, the observed daily counts consist of $D$ time series $\lbrace Z^{(d)}_t, \, t \in [1, T], \,  d \in [1, D] \rbrace$, with $[1, T]$ the integers $\lbrace 1, 2, \hdots, T\rbrace$. 
The pandemic spread on each territory $d$ is thus characterized by both a reproduction number time series and an \textit{outlier} time series $\lbrace (R_t^{(d)},O_t^{(d)}), \, t \in [1, T] \rbrace$, 
that must thus be estimated jointly for all territories.

These considerations lead to a first key contribution of the present work:  the modeling of corrupted daily infection counts across territories by extending Model~\eqref{equ:cori} to multivariate daily new infection Poisson counts, each with an instantaneous Poisson parameter $p_t^{(d)}$,  
with $(\Phi Z)_t^{(d)} = \sum_{u = 1}^{\tau_\Phi} \Phi_u Z_{t-u}^{(d)}$: 
\begin{align}
\label{equ:coriO_spat}
 \PP & \left(Z_t^{(d)} \vert Z_{1:t-1}^{(d)}; R_{t}^{(d)}, O_{t}^{(d)}\right)\\
\nonumber & \sim \mathrm{Poiss}\left(p^{(d)}_t \triangleq  R_t^{(d)} \times (\Phi Z)_t^{(d)}  + O_t^{(d)} \right).
\end{align}

\subsection{Crafting the convex nonsmooth estimation functional} 
\label{ssec:multiple}

The goal is now to estimate jointly  $R = \lbrace R_t^{(d)}, \, t \in [1, T],  d \in [1, D] \rbrace$ and $O = \lbrace O_t^{(d)}, \, t \in [1, T],  d \in [1, D] \rbrace$ from $\lbrace Z_t^{(d)}, \, t \in [1, T], \,  d \in [1, D] \rbrace$. 
The next original contribution is thus to frame the estimation into an inverse problem strategy, with a careful crafting of the functional to minimize combining the negative log-likelihood of \textit{extended} Model~\eqref{equ:coriO_spat}
with positiveness and regularity constraints in time and space on $R$ and structure on $O$. 

\subsubsection{Negative log-likelihood function} 
\label{ssec:}

The negative log-likelihood associated to the \textit{extended} Poisson Model~\eqref{equ:coriO_spat} is defined in a separable manner as 
\begin{align}
\label{eq:DKL_outlier}
&F(R, O \lvert Z) \triangleq \sum_{d = 1}^{D} \sum_{t = 1}^T f(R_t^{(d)}, O_t^{(d)} \lvert Z_t^{(d)}, (\Phi Z)_t^{(d)}), \\
&\text{with} \, \,  \, f(r, o \lvert z, (\Phi z)) \triangleq 
\begin{cases}
\iota_{\lbrace (0,0) \rbrace}(r, o) \, \, \, \text{if} \,\, z =  \Phi z  = 0, \\
d_{\mathtt{KL}}(z \lvert r \times (\Phi z) + o) \, \, \,\text{otherwise.} \nonumber
\end{cases}
\end{align}
The indicator function of the singleton $\lbrace (0,0) \rbrace$ permits to force the (MLE) estimates of $R$ and $O$ to vanish when the pandemic stops, i.e.,  when $\left(\Phi Z\right)_{t}^{(d)} =Z_t^{(d)} = 0$,
\begin{align}
\label{eq:def_indicator}
\iota_{\lbrace (0,0) \rbrace}(r, o) = \begin{cases}
0 & \text{if} \quad r =o = 0 \\
\infty & \text{otherwise,}
\end{cases}
\end{align}
while the standard Kullback-Leibler divergence (KLD) is defined as $D_{\mathtt{KL}}(Z\lvert P) = \sum_{t = 1}^T\sum_{d=1}^D d_{\mathtt{KL}}(Z_t^{(d)} \lvert P_t^{(d)})$, with
\begin{align}
\label{eq:def_DKL}
d_{\mathtt{KL}}(z \lvert p) \triangleq &\begin{cases}
z \ln \frac{z}{p} + p - z \quad &\text{if} \, \,  z > 0,  \,p > 0\\
p \quad &\text{if} \, \, z = 0, \,  p \geq 0\\
\infty \quad &\text{else.} 
\end{cases}
\end{align}

\subsubsection{Regularity and positivity constraints on $R$} 
\label{ssec:constr_R}

Unlike Model~\eqref{equ:cori}, from which a unique MLE estimate can be derived (Eq.~\eqref{eq:CoriMLE}), for Model~\eqref{equ:coriO_spat} the \textit{extended} likelihood (Eq.~\eqref{eq:DKL_outlier}) is no longer strictly convex, thus does not have a \textit{unique} maximum, and hence requires additional constraints on the estimates of $R$ and $O$.

\noindent \textit{i) Time regularity:}  By nature and to be of any potential use for actual pandemic monitoring, the estimated reproduction numbers for a given territory must vary only slowly and smoothly along time. 
Following~\citep{abry2020spatial}, it has been chosen to enforce piecewise linear time evolution of the estimate of the reproduction numbers. 
Such a property ensures that, besides an estimate of $R$ at day $t$, one also gets an estimation of a local trend, assessing whether the pandemic is growing or decreasing. 
Following~\citep{colas2019,Debarre:2020vi},  to favor piecewise linear temporal evolution of the estimates, we penalize the $\ell_1$-norm of the multivariate time-domain Laplacian $\boldsymbol{\mathrm{D}}_2 : \mathbb{R}^{D \times T}  \rightarrow  \mathbb{R}^{D \times (T-2)}$ of the estimates of $R$:
\begin{align}
\label{eq:cons_R_spat}
\lVert \boldsymbol{\mathrm{D}}_2 R\rVert_1 \triangleq \sum_{d = 1}^{D} \sum_{t = 2}^{T-1}\lvert (\boldsymbol{\mathrm{D}}_2 R)_t^{(d)} \rvert,
\end{align}
with $(\boldsymbol{\mathrm{D}}_2 R)_t^{(d)} \triangleq \frac{1}{2}  R^{(d)}_{t-1} - R^{(d)}_t + \frac{1}{2} R^{(d)}_{t+1}$.

\noindent \textit{ii) Positivity:}  $R$ is by nature positive.
Yet, Eq.~\eqref{eq:DKL_outlier} indicates that, per se, the proposed \textit{extended} negative log-likelihood ensures the positivity of $p_t^{(d)} = R_t^{(d)} \times (\Phi Z)_{t}^{(d)}  + O_t^{(d)}$ but not of $R_t^{(d)}$ itself. 
Thus, to enforce the positivity of $R_t^{(d)}$, an indicator function is added in the constraints:
\begin{align}
\label{eq:cons_R_spat_iota}
\iota_{\geq 0}(R)
\triangleq \left\lbrace 
\begin{array}{cc}
\infty & \text{if} \,  \exists \, t \in [1, T], \, d \in [1, D], \,  \text{s.t.} \, R^{(d)}_t < 0\\
0 & \text{otherwise.}
\end{array}
\right. 
\end{align}

\noindent \textit{iii) Spatial regularity:}  When different territories $d$ are connected, it is natural to expect that their pandemic dynamics are correlated and thus that the corresponding estimates of $R$ show some \textit{spatial} regularity.
Following~\citep{abry2020spatial}, we use a graph, with $D$ vertices corresponding to the different territories and $E$ edges to connections between pairs of territories.
We favor piecewise constant spatial estimates of $\hat{R}_t^{1:D}$ across territories by penalizing the \textit{Graph Total Variation}, defined as: 
\begin{align}
\label{eq:def_GTV}
\mathrm{GTV}(R) =\sum_{t = 1}^T \sum_{d_1 \sim d_2}  \left\lvert R_t^{(d_1)} - R_t^{(d_2)} \right\rvert,
\end{align}
where $ \sum_{d_1 \sim d_2} $ runs over vertices connected by an edge.

The \textit{Graph Total Variation} can be rewritten as a composition of a discrete difference operator and a sparsity-promoting $\ell_1$-norm, $\mathrm{GTV}(R)  = \lVert \boldsymbol{\mathrm{G}} R \rVert_1$, defining  the linear operator $\boldsymbol{\mathrm{G}}$:
\begin{align}
\label{eq:def_G_spat}
\boldsymbol{\mathrm{G}} : \left\lbrace
\begin{array}{ccc}
\mathbb{R}^{D \times T} & \rightarrow & \mathbb{R}^{E \times T} \\
R & \mapsto & \left\lbrace R_t^{(d_1)} - R_t^{(d_2)}, \, t \in [1, T], \, d_1 \sim d_2 \right\rbrace.
\end{array}
\right.
\end{align}

\subsubsection{Constraints on the \textit{outliers}}
\label{ssec:constr_O}
The discussion in Section~\ref{sec-outlier} suggests that \textit{outliers} have a sparse structure: 
They occur with significant values only on specific days (sundays, day-offs,\ldots) and are essentially negligible elsewhere.   
Sparsity is classically enforced in inverse problem settings by constraining the $\ell_1$-norm.
We thus further propose to impose the following constraint to the estimates of the \textit{outliers}:
\begin{align}
\label{eq:cons_O}
\Vert O \Vert_1 \triangleq \sum_{d = 1}^{D}\sum_{t = 1}^T \lvert O_t^{(d)} \rvert.
\end{align}

\subsubsection{Estimation as an optimization problem}
\label{ssec:}

Combining the data-fidelity term of Eq.~\eqref{eq:DKL_outlier} and the penalizations of Eqs.~\eqref{eq:cons_R_spat}, \eqref{eq:cons_R_spat_iota}, \eqref{eq:def_GTV}, and \eqref{eq:cons_O}, we propose to obtain the estimates of $R$ and $O$ by minimizing:
\begin{align}
\label{eq:penal_KL}
\nonumber (\hat{R}, \hat{O}) \in \underset{R, O \in \mathbb{R}^{D \times T}}{\mathrm{Argmin}}\, &F(R, O \lvert Z) + \lambda_{\mathrm{T}}\lVert \boldsymbol{\mathrm{D}}_2 R \rVert_1 + \iota_{\geq 0}(R) \\
&+ \lambda_{\mathrm{S}} \lVert \boldsymbol{\mathrm{G}} R \rVert_1 + \lambda_{\mathrm{O}} \Vert O \Vert_1, 
\end{align}
with $\lambda_{\mathrm{T}} > 0$, $\lambda_{\mathrm{S}} > 0$ and $\lambda_{\mathrm{O}} > 0$, \textit{regularization hyperparameters}, balancing the strengths of the different constraints one against the others and against the likelihood. 

\subsubsection{Reformulation of the optimization problem}
\label{ssec:reform}
To study the theoretical properties of Problem~\eqref{eq:penal_KL}, and to derive a minimization algorithm, it is useful to recast Eq.~\eqref{eq:penal_KL} into a generic formulation:
\begin{align}
\label{eq:gen_var}
( \hat{R}, \hat{O}) = \underset{R, O \in \mathbb{R}^{D \times T}}{\mathrm{Argmin}} \, F(R,O \lvert Z) + H(\boldsymbol{\mathrm{L}}(R, O)), 
\end{align}
with 
$\boldsymbol{\mathrm{L}}$ a linear operator and $H$ a function enforcing constraints  promoting certain sparsity / non-negativity, defined as:
\begin{align}
\label{eq:def_L_spat}
\boldsymbol{\mathrm{L}} :  \left\lbrace
\begin{array}{ccc}
\left( \mathbb{R}^{D \times T} \right)^2 &  \rightarrow  & \mathbb{R}^{D \times (T-2)} \times \mathbb{R}^{D \times T} \times \mathbb{R}^{E \times T} \times \mathbb{R}^{D \times T} \\
(R,O)  & \mapsto & \left(   \lambda_{\mathrm{T}}\boldsymbol{\mathrm{D}}_2 R, R, \lambda_{\mathrm{S}} \boldsymbol{\mathrm{G}} R, \lambda_{\mathrm{O}} O \right),
 \end{array}
 \right.
\end{align}
\begin{align}
\label{eq:def_H_spat}
H(Q_1, Q_2, Q_3, Q_4 ) \triangleq  \lVert Q_1 \rVert_1 + \iota_{\geq 0}(Q_2) + \lVert Q_3 \rVert_1+ \lVert Q_4 \rVert_1.
\end{align}

\subsection{Convexity, existence of a minimizer and uniqueness} 
\label{ssec:convex}

Showing that Eq.~\eqref{eq:penal_KL} provides well-defined estimates of $R$ and $O$ requires to prove the existence of a minimizer of $F(\cdot) + H(\boldsymbol{\mathrm{L}} \cdot)$. 
To that aim, preliminary results are necessary.

\begin{lemma}
\label{lem:F_convex}
The extended Kullback-Leibler data-fidelity term $F$ (Eq.~\eqref{eq:DKL_outlier}) is jointly convex with respect to the pair $(R, O)$.
\end{lemma}

\begin{proof}
$F$ being fully separable in time and space, it is sufficient to demonstrate the joint convexity of $f(r, o \lvert z, (\Phi z)) $ with respect to the pair $(r, o)$. 
If $z = (\Phi z) = 0$, the convexity of $f(r, o \lvert z, (\Phi z)) $ w.r.t. $(r, o)$ reduces to the convexity of the indicator function of the convex set $\lbrace (0,0)\rbrace$. 
Else, since the one-component KLD $d_{\mathtt{KL}}(z \lvert p)$ is convex with respect to its second argument~\citep{combettes2007douglas}, and $(r,o) \mapsto r \times (\Phi z)+o $ being linear, $(r,o) \mapsto f(r, o \lvert z, (\Phi z))$ is convex.
\end{proof}

\begin{corollary}
\label{cor:convex}
The penalized Kullback-Leibler functional (Eq.~\eqref{eq:penal_KL}) is jointly convex with respect to the pair $(R,O)$.
\end{corollary}

\begin{proof}
First, Lemma~\ref{lem:F_convex} ensures the convexity of the data-fidelity term $F$.
Then, the function $H$, appearing in 
Eq.~\eqref{eq:gen_var}, involves $\ell_1$-norms and an indicator function which are convex, thus it is convex and the composition of $H$ with the linear operator $\boldsymbol{\mathrm{L}}$ is convex as well.
As the sum of two convex terms, the penalized Kullback-Leibler functional is convex.
\end{proof}

\noindent Further, because the Kullback-Leibler data-fidelity term $F$ (Eq.~\eqref{eq:DKL_outlier}), as well as the indicator function in Eq.~\eqref{eq:cons_R_spat_iota}, may take infinite values, the domain of the convex functional in Eq.~\eqref{eq:penal_KL} is not the whole space $ \left(\mathbb{R}^{D \times T} \right)^2$, but restricts to a domain $\Omega$ of \textit{feasible} variables, described in Lemma~\ref{lem:domain_D}.

\begin{lemma}
\label{lem:domain_D}
The domain $\Omega$ of \textit{feasible} points for Problem~\eqref{eq:penal_KL} consists of pairs $(R, O)$ such that, $\forall u \in [ 1, t]$, $\forall  d \in [1, D]$
\begin{itemize}
\item[\textit{(i)}]  $R_t^{(d)} (\Phi Z )_t^{(d)} + O_t^{(d)} \geq 0$,
\item[\textit{(ii)}] if $Z_t^{(d)} = \left(\Phi Z\right)_t^{(d)} = 0$ then  $ R_t^{(d)} = O_t^{(d)} = 0$,
\item[\textit{(iii)}] $R_t^{(d)} \geq 0$.
\end{itemize}
It is convex and closed.
\end{lemma}

\begin{proof}
Conditions \textit{(i)} and \textit{(ii)} reflect the fact that the \textit{extended} Kullback-Leibler data-fidelity term $F$ (Eq.~\eqref{eq:DKL_outlier}) must be finite for feasible time series, while Condition~\textit{(iii)} stems from the positivity constraint enforced by the indicator function in Eq.~\eqref{eq:cons_R_spat_iota}.
Further, each Condition \textit{(i)} to \textit{(iii)} defines a closed convex set of feasible time series, the finite intersection of closed convex sets being closed and convex as well, $\Omega$ is a closed convex set.
\end{proof}

\noindent 
Then, the existence of a minimizer can be proven.

\begin{theorem}
\label{thm:minimizer}
The penalized Kullback-Leibler functional in Eq.~\eqref{eq:penal_KL} is lower bounded.
Further, 
if $(\Phi Z)_t^{(d)} > 0$ for all $t \in [1, T]$, $ d \in [1, D]$, then
the minimization problem of Eq.~\eqref{eq:penal_KL} has, a least, one solution.
\end{theorem}

\begin{proof}
Since, for feasible points, the standard Kullback-Leibler divergence, the indicator functions and the $\ell_1$-norms are nonnegative, the functional is lower bounded by zero.\\
Let $\left\lVert (R, O) \right\rVert_2 \triangleq \sqrt{\lVert R \rVert_2^2 + \lVert O \rVert_2^2}$ be the canonical norm on the product space $\left(\mathbb{R}^{D\times T} \right)^2$.
Let $\boldsymbol{1}$ (resp.  $\boldsymbol{0}$) the multivariate time series with all entries equal to 1 (resp. 0),  $\left( \boldsymbol{1}, \boldsymbol{0}\right) \in \Omega$ is a feasible point. 
Let 
\begin{align}
\mu \triangleq F(\boldsymbol{1},\boldsymbol{0} \lvert Z) + H(\boldsymbol{\mathrm{L}}(\boldsymbol{1},\boldsymbol{0})) = F(\boldsymbol{1},\boldsymbol{0} \lvert Z)\geq 0
\end{align}
As shown below,  the functional $F(\cdot \lvert Z) + H(\textbf{L} \cdot)$ is \textit{coercive},  
thus there exists $\alpha > 0$ such that
\begin{align}
\label{eq:out_Dalpha}
\left\lVert (R, O) \right\rVert_2 > \alpha \Rightarrow F(R,O \lvert Z) + H(\boldsymbol{\mathrm{L}}(R,O)) > \mu.
\end{align} 
Let $\Delta_{\alpha}$ be the closed ball of radius $\alpha$ in $\left( \mathbb{R}^{D\times T}\right)^2$ for the canonical norm $\lVert \cdot \rVert_2$.
Since the feasible set $\Omega$ is closed, the intersection $\Omega \cap \Delta_{\alpha}$ is bounded and closed, hence compact as a subset of a finite-dimensional Hilbert space.
Further, the functional $F(\cdot \lvert Z) + H(\textbf{L} \cdot)$ being continuous over $\Omega \cap \Delta_{\alpha}$ it reaches its minimum over this set at some point $( \hat{R},  \hat{O} ) \in \Omega \cap \Delta_{\alpha}$.
Then, since $\left(\boldsymbol{1},\boldsymbol{0}\right) \in \Omega\cap \Delta_{\alpha}$, it follows that $F(\hat{R},\hat{O} \lvert Z) + H(\boldsymbol{\mathrm{L}}(\hat{R},\hat{O})) \leq \mu$. 
From Eq.~\eqref{eq:out_Dalpha}, if $\left(R,O\right) \in \Omega \setminus \Delta_{\alpha}$, then one has $F(R,O \lvert Z) + H(\boldsymbol{\mathrm{L}}(R,O)) > \mu$, 
showing that the relative minimum $( \hat{R},  \hat{O} )$ on $\Omega\cap \Delta_{\alpha}$ is to be a global minimum on the whole domain of feasible points $\Omega$. To conclude the proof, we need to compute $\alpha$ satisfying~\eqref{eq:out_Dalpha}, which follows from:  
\begin{itemize}
\item[\textit{i)}] If $\lVert O\lVert_1 > C \triangleq \mu/\lambda_{\mathrm{O}}$,  then $H(\boldsymbol{\mathrm{L}}(R,O)) \geq \lambda_{\mathrm{O}} \lVert O\lVert_1 > \mu$.
\item[\textit{ii)}] Since $\forall z \geq 0$, $\lim_{p \rightarrow \infty} d_{\mathtt{KL}}(z\lvert p) = \infty$ there is a constant $c = c(Z)$ such that $D_{\mathtt{KL}}(Z\lvert Q) > \mu$ as soon as $\lVert Q\lVert_\infty > c$.
Therefore, with $c' \triangleq \min_{u,d} (\Phi Z)_t^{(d)}$, as soon as $\lVert O\lVert_1 \leq C$ and $\lVert R\lVert_\infty > (c+C)/c'$, we have $\lVert R \times (\Phi Z) + O\lVert_\infty \geq  c' \lVert R \lVert_\infty-\lVert O\lVert_\infty \geq c' \lVert R\lVert_\infty - \lVert O\lVert_1 > c$, and hence $F(R,O\lvert Z) \geq \mu$.
\end{itemize}
In finite dimension, the equivalence of the norms yields constants $C_{\mathrm{O}}$ and $C_{\mathrm{R}}$ such that, $\lVert O \rVert_2 \geq C_{\mathrm{O}} \Rightarrow \lVert O \rVert_1 \geq C$ and $\lVert R \rVert_2 \geq C_{\mathrm{R}} \Rightarrow \lVert R \rVert_{\infty} \geq (c+C)/c'$.
Finally, $\alpha \triangleq \sqrt{C_{\mathrm{O}}^2 + C_{\mathrm{R}}^2}$ satisfies~\eqref{eq:out_Dalpha}.
\end{proof}

\noindent The regularizing functional in Eq.~\eqref{eq:penal_KL} not showing strict convexity, investigating the uniqueness of its minimizer requires further analysis, beyond the scope of this work.
Nevertheless, leveraging the strict convexity of the standard KLD in Eq.~\eqref{eq:DKL_outlier},  a uniqueness property for the estimate of the Poisson parameter $p_t^{(d)}$ is derived in Proposition~\ref{prop:str_conv},  based on Lemma~\ref{lem:min_gen_str}.

\begin{lemma}[\citep{Bauschke:2011ta}, \textit{Chapter 11}]
\label{lem:min_gen_str}
Let $g : \mathbb{R}^n \rightarrow \mathbb{R}$, $h : \mathbb{R}^m \rightarrow\mathbb{R}$ two proper, lower-semicontinuous, convex functions and $\boldsymbol{\mathrm{B}} : \mathbb{R}^m \rightarrow \mathbb{R}^n$ a linear operator.
Consider the following convex minimization problem
\begin{align}
\label{eq:min_gen_str}
\hat{x} \in \underset{x \in \mathbb{R}^m}{\mathrm{Argmin}} \, g(\boldsymbol{\mathrm{B}} x ) + h( x).
\end{align}
If $g$ is strictly convex, then there exists $\hat{y}\in \mathbb{R}^n$, such that for \textit{any} minimizer $\hat{x}$ of~\eqref{eq:min_gen_str}, $\hat{y} = \boldsymbol{\mathrm{B}} \hat{x}$.
\end{lemma}

\begin{proposition}
\label{prop:str_conv}
Let $( \hat{R}, \hat{O})$ denote a solution of Problem~\eqref{eq:penal_KL}.
Let $\hat{P}  \triangleq  \{ \hat p_t^{(d)}, t\in \{1, \ldots T \}, \, d \in \{1,\hdots D \}\}$, then $\hat{P} $ is \textit{unique}.
\end{proposition}

\begin{proof}
In the trivial situation when the daily counts are uniformly vanishing, i.e., $Z = \Phi Z  \equiv 0$, the \textit{unique} minimizer coincides with the \textit{unique} feasible point and consists in estimated reproduction numbers and \textit{outliers} equal to zero at any time and in any territory. \\
Otherwise,  two situations are to be considered
\begin{itemize}
\item[\textit{(i)}] \textit{If $Z$ and $ \Phi Z $ never vanish simultaneously:} 
In this case $F(R,O \lvert Z) = D_{\mathtt{KL}}(Z\lvert \boldsymbol{\mathrm{B}}(R,O))$,  with $\boldsymbol{\mathrm{B}} : ( R, O) \mapsto R \times (\Phi Z) + O$.
Further, the standard Kullback-Leibler divergence $g(P) =  D_{\mathtt{KL}}(Z\lvert P)$ is strictly convex. 
Hence,  setting $h = H(\boldsymbol{\mathrm{L}}\cdot)$, Proposition~\ref{prop:str_conv} yields the result.
\item[\textit{(ii)}] \textit{If there exist some times $u$ and territories $d$ s.t. $Z_t^{(d)} = (\Phi Z )_t^{(d)} = 0$:} For these times and territories, the estimated instant Poisson parameter is forced to cancel by the indicator function involved in Eq.~\eqref{eq:DKL_outlier}.
Then, restricting to the indices for which $Z$ and $\Phi Z$ are not vanishing simultaneously,  the data-fidelity term $F$ is strictly convex and the uniqueness of $\hat{P}$ follows 
by the same reasoning.
\end{itemize}
\end{proof}

\subsection{Minimization \textit{via} a proximal algorithm} 
\label{ssec:min_prox}

From Corollary~\ref{cor:convex}, the penalized Kullback-Leibler functional is convex, and, from Theorem~\ref{thm:minimizer}, it has at least one minimizer. 
The $\ell_1$-norm and the indicator functions in the objective function make the overall functional non-smooth, preventing from the use of standard \textit{gradient} descent.
Minimizing~\eqref{eq:penal_KL} thus requires more advanced tools, handling nondifferentiable functions, such as \textit{proximal} algorithms.

\subsubsection{Adjoint and proximal operators} 
Let $\boldsymbol{\mathrm{L}}$ a  bounded linear operator $\boldsymbol{\mathrm{L}}$,  its \textit{adjoint} operator is denoted $\boldsymbol{\mathrm{L}}^*$ and its operator norm is $\lVert \boldsymbol{\mathrm{L}} \rVert_{\mathrm{op}} = \sup\lbrace \lVert \boldsymbol{\mathrm{L}}X \rVert_2, \, \lVert X \rvert_2 \leq 1\rbrace$.
The \textit{convex conjugate} of the proper, lower-semicontinuous function $H : \mathcal{X} \rightarrow\mathbb{R}$, is defined as $H^*(Y) = \sup_{X \in \mathcal{X}} \, \langle Y, X\rangle - H(X)$~\citep{bauschke2011convex}, with  $\mathcal{X}$ a Hilbert space.
For $\tau \in \RR^+ $ 
the \textit{proximal operator} of $\tau H$ at $Y \in \mathcal{X}$ reads by definition: $\mathrm{prox}_{\tau H } (Y) \triangleq  \arg\min_{X \in \mathcal{X}} \, 1/2 \lVert X-Y \rVert^2 + \tau H(X)$.

\subsubsection{Primal-dual algorithmic scheme}

In this work, 
we chose to implement the Chambolle-Pock primal-dual minimization scheme~\citep{chambolle2011first}.
Making use of the compact reformulation in Eq.~\eqref{eq:gen_var}, we perform here an explicit particularization of the Chambolle-Pock algorithm detailed in Algorithm~\ref{alg:CP}.
From~\citep{chambolle2011first}, Algorithm~\ref{alg:CP} converges toward a minimizer of Problem~\eqref{eq:gen_var}, if the descent parameters $\tau$, $\sigma$ satisfy
\begin{align}
\label{eq:cond_cvPD}
\tau \sigma \lVert \boldsymbol{\mathrm{L}} \rVert_{\mathrm{op}}^2 < 1.
\end{align}
Therefore, the actual and practical use of Algorithm~\ref{alg:CP} and enforcement of Condition~\eqref{eq:cond_cvPD}, 
requires to evaluate, at least,  an \textit{upper bound} of $\lVert \boldsymbol{\mathrm{L}} \rVert_{\mathrm{op}}$.

\begin{proposition}
\label{prop:up_normL}
Let $\boldsymbol{\mathrm{L}}$ be defined in Eq.~\eqref{eq:def_L_spat}, then
\begin{align}
\lVert \boldsymbol{\mathrm{L}} \rVert_{\mathrm{op}}^2 \leq \max \lbrace \lambda_{\mathrm{T}}^2 \lVert \boldsymbol{\mathrm{D}}_2 \rVert_{\mathrm{op}}^2 + \lambda_{\mathrm{S}} \lVert \boldsymbol{\mathrm{G}}\rVert_{\mathrm{op}}^2 + 1, \lambda_{\mathrm{O}}^2 \rbrace.
\end{align}
\end{proposition}
\begin{proof}
By definition of the operator norm $\lVert \boldsymbol{\cdot} \rVert_{\mathrm{op}}$, one has
\begin{align}
&\lVert\boldsymbol{\mathrm{L}}(R,O )\rVert_2^2\\
&= \lambda_{\mathrm{T}}^2 \lVert \boldsymbol{\mathrm{D}}_2 R \rVert_2^2 + \lVert R \rVert_2^2+ \lambda_{\mathrm{S}}^2 \lVert \boldsymbol{\mathrm{G}} R \rVert_2^2+\lambda_{\mathrm{O}}^2 \lVert O \rVert_2^2 \nonumber \\
& \leq  ( \lambda_{\mathrm{T}}^2\lVert \boldsymbol{\mathrm{D}}_2 \rVert_{\mathrm{op}}^2 + 1 +  \lambda_{\mathrm{S}}^2\lVert \boldsymbol{\mathrm{G}} \rVert_{\mathrm{op}}^2)\lVert R \rVert_2^2+\lambda_{\mathrm{O}}^2 \lVert O \rVert_2^2\nonumber\\
&\leq \max \lbrace \lambda_{\mathrm{T}}^2 \lVert \boldsymbol{\mathrm{D}}_2 \rVert_{\mathrm{op}}^2 + 1 + \lambda_{\mathrm{S}}^2 \lVert \boldsymbol{\mathrm{G}} \rVert_{\mathrm{op}}^2, \lambda_{\mathrm{O}}^2 \rbrace \lVert (R,O) \rVert_2^2 .\nonumber 
\end{align}
\end{proof}
In practice,  equal descent step sizes are chosen,  so that they saturate Condition~\eqref{eq:cond_cvPD}, replacing the unknown $\lVert \boldsymbol{\mathrm{L}} \rVert_{\mathrm{op}}$ by its upper bound derived at Proposition~\ref{prop:up_normL}, leading to
\begin{align}
\tau = \sigma = \frac{0.99}{\sqrt{\max \lbrace \lambda_{\mathrm{T}}^2 \lVert \boldsymbol{\mathrm{D}}_2 \rVert_{\mathrm{op}}^2 + \lambda_{\mathrm{S}} \lVert \boldsymbol{\mathrm{G}}\rVert_{\mathrm{op}}^2 + 1, \lambda_{\mathrm{O}}^2 \rbrace}}.
\end{align}

\begin{algorithm}[t!]
\caption{Primal-dual minimization of the penalized Kullback-Leibler~\eqref{eq:penal_KL} for the estimation of reproduction numbers and \textit{outliers}.
\label{alg:CP}}
\begin{algorithmic}
\REQUIRE	Infection counts: ${Z} \in \mathbb{R}^{D\times T}$ and $(\Phi Z)\in \mathbb{R}^{D\times T}$
\STATE \textbf{Choose} descent parameters: $\tau, \sigma > 0$,  precision: $\varepsilon$
\STATE Max.  iterations: $k_{\max}$, length of window: $k_{\mathrm{smooth}}$
\PRINT $R^{[0]} = {Z}$,  $O^{[0]} = \boldsymbol{0}$, 
\STATE $Q^{[0]} = \boldsymbol{\mathrm{L}} (R^{[0]}, O^{[0]})$,  $\overline{R}^{[0]} = R^{[0]}$,  $\overline{O}^{[0]} = O^{[0]}$
\STATE $ \Phi_{0} = F(R^{[0]},O^{[0]} \lvert Z) + H(\boldsymbol{\mathrm{L}}(R^{[0]}, O^{[0]}) $
\STATE $\Psi_0 = \varepsilon + 1$, $\overline{\Psi}_0 = \varepsilon + 1$, $k = 0$
\WHILE{$\Psi_k \geq \epsilon$ and $k < k_{\max}$}
\STATE {\color{gray}Update the dual, primal and auxiliary variables}
\begin{align*}
Q^{[k+1]} &= \mathrm{prox}_{\sigma H^*}(Q^{[k]} + \sigma \boldsymbol{\mathrm{L}} (\overline{R}^{[k]},\overline{O}^{[k]}))\\
(R^{[k+1]}, O^{[k+1]} )  &=  \mathrm{prox}_{\tau F}((R^{[k]}, O^{[k]} ) - \tau \boldsymbol{\mathrm{L}}^* Q^{[k+1]})\\
 (\overline{R}^{[k+1]}, \overline{O}^{[k+1]} )& = 2(R^{[k+1]}, O^{[k+1]} ) - (R^{[k]}, O^{[k]} )
\end{align*}
\vspace{-5mm}
\STATE {\color{gray}Computation of the convergence criteria\color{gray}}
\begin{align}
\Phi_{k+1} &= F(R^{[k+1]},O^{[k+1]} \lvert Z) + H(\boldsymbol{\mathrm{L}}(R^{[k+1]}; O^{[k+1]})) \nonumber\\
 \Psi_{k+1} &=  \lvert \Phi_{k+1} - \Phi_k \rvert / \Phi_k  \nonumber \\
\label{eq:slid_max}\overline{\Psi}_{k+1} &= \max \left\lbrace \Psi_{\ell} \, : \,\lvert k - k_{\mathrm{smooth}} \rvert_+ \leq \ell \leq  k+1 \right\rbrace  
\end{align}
\vspace{-4mm}
\STATE $k  \leftarrow k+1$
\ENDWHILE
\end{algorithmic}
\end{algorithm}

\subsubsection{Closed-form expression of the proximal and adjoint operators}

For proximal algorithms, such as Algorithm~\ref{alg:CP}, the key ingredient to achieve an actual and fast implementation is to derive a closed-form expression for each proximal operator involved. 

\noindent \textit{Proximal operator of the regularizing function:} The dual variable update in Algorithm~\ref{alg:CP} involves the proximal operator of the convex conjugate of the penalty $H$ defined  in~\eqref{eq:def_H_spat}.
In practice, first, the proximal operator of $H$ is computed, and then, the proximal operator ofits convex conjugate $H^*$ is derived applying Moreau's identity, which states that, for $\sigma > 0$, $ \mathrm{prox}_{\sigma H^*}(X) =X - \sigma \mathrm{prox}_{H/\sigma}(X/\sigma)$. 
Then, function $H/\sigma$ being fully separable, both over its four entries and over the components of each entry, its proximal operator can be computed component-wise. 
These computations involve two standard closed-form proximal operators: \textit{i)} the proximal operator of the absolute value, known as the \textit{soft-thresholding}, $\mathrm{prox}_{\sigma \lvert \boldsymbol{\cdot} \rvert} ( q) = \max \left\lbrace \lvert q \rvert - \sigma, 0\right\rbrace q/ \lvert q \rvert$, and \textit{ii)}~the proximal operator of the indicator function, which amounts to a projection: $\mathrm{prox}_{ \iota_{\geq 0}} ( q) = \max \left\lbrace 0, q \right\rbrace$.
 
\noindent \textit{Adjoint of the linear operator:} Updating the primal variable requires first to compute $\boldsymbol{\mathrm{L}}^* : \mathbb{R}^{D \times (T-2)} \times \mathbb{R}^{D \times T} \times \mathbb{R}^{E \times T} \times \mathbb{R}^{D \times T}  \rightarrow  \left( \mathbb{R}^{D \times T} \right)^2  $,  which can be expressed in terms of the adjoint operators of $\boldsymbol{\mathrm{D}}_2$ and $\boldsymbol{\mathrm{G}}$, \textit{via}
\begin{align}
\boldsymbol{\mathrm{L}}^* :
(Q_1,Q_2,Q_3,Q_4)   \mapsto  ( \lambda_{\mathrm{T}}\boldsymbol{\mathrm{D}}^*_2 Q_1+ Q_2 + \lambda_{\mathrm{S}} \boldsymbol{\mathrm{G}}^* Q_3 ,\lambda_{\mathrm{O}} Q_4).
\end{align}

\noindent \textit{Proximal operator of the data-fidelity term:} The \textit{extended} negative log-likelihood $F$ of Eq.~\eqref{eq:DKL_outlier} being separable over the times $t \in [1, T]$ and territories $ d \in [1, D]$, its proximal operator can be computed in a component-wise manner.
Each term in~\eqref{eq:DKL_outlier} corresponds to the composition of a linear map with the standard Kullback-Leibler divergence $d_{\mathtt{KL}}(z\lvert\cdot)$, whose proximal operator writes explicitly~\citep{combettes2007douglas}:
\begin{align}
\mathrm{prox}_{\tau d_{\mathtt{KL}}(z \lvert \cdot )}(p) = ( p- \tau + \sqrt{(p-\tau)^2 + 4 \tau z})/2.
\end{align}
\begin{proposition}
For $\beta \triangleq (\Phi z)^2 + 1$,  $\textbf{a} : (r,o) \mapsto r \times (\Phi z) + o$,  
\begin{align}
\label{eq:prox_d_KL}
&\mathrm{prox}_{\tau f(\boldsymbol{\cdot}, \boldsymbol{\cdot}  \lvert z, (\Phi z))}(r, o) \\
\nonumber &= 
\left\lbrace
\begin{array}{l}
(0,0), \hfill \text{if} \, \,  z =  (\Phi z)  = 0\\
(r,o)- \textbf{a}^* \beta^{-1}  \left( \textbf{I} - \mathrm{prox}_{\tau \beta d_{\mathtt{KL}}(z \lvert \boldsymbol{\cdot}) }\right) \textbf{a} (r,o), \, \,  \, \text{otherwise}.
\end{array}
\right.
\end{align}
from which derives straightforwardly the proximal operator of the global data-fidelity function $F$.
\end{proposition}
\begin{proof}
Since $\textbf{a} \textbf{a}^* =  \beta >0$, the formula in Eq.~\eqref{eq:prox_d_KL} is obtained applying Proposition 3.4 of~\citep{pustelnik2011parallel}.
\end{proof}

\subsubsection{Stopping criterion}

The final step in the actual of an iterative algorithm is to choose a stopping criterion.
Classically, the convergence of Algorithm~\ref{alg:CP} is quantified by the decrease of the normalized increments of the objective functional.
Fig.~\ref{fig:conv_crit} shows that such increments (dashed blue curves) display large and erractic fluctuations, thus precluding their use to construct a reliable stopping criterion. 
Instead, and to provide a robust convergence criterion, we propose to make use of a nonlinearly smoothed normalized increments (solid black curve in Fig.~\ref{fig:conv_crit}), obtained as a maximum over a sliding window over the past $k_{\mathrm{smooth}}$ iterations, as detailed in Eq.~\eqref{eq:slid_max}.
Hence, in our concrete experiments,  we set $k_{\mathrm{smooth}} = 500$ and Algorithm~\ref{alg:CP} is stopped whenever the smoothed increment reaches a precision $\varepsilon$, set to $10^{-7}$, or when the number of iterations exceeds a maximum of $k_{\max}$ iterations, set to $10^7$.

\begin{figure}[h!]
\centerline{\includegraphics[width = .4\linewidth]{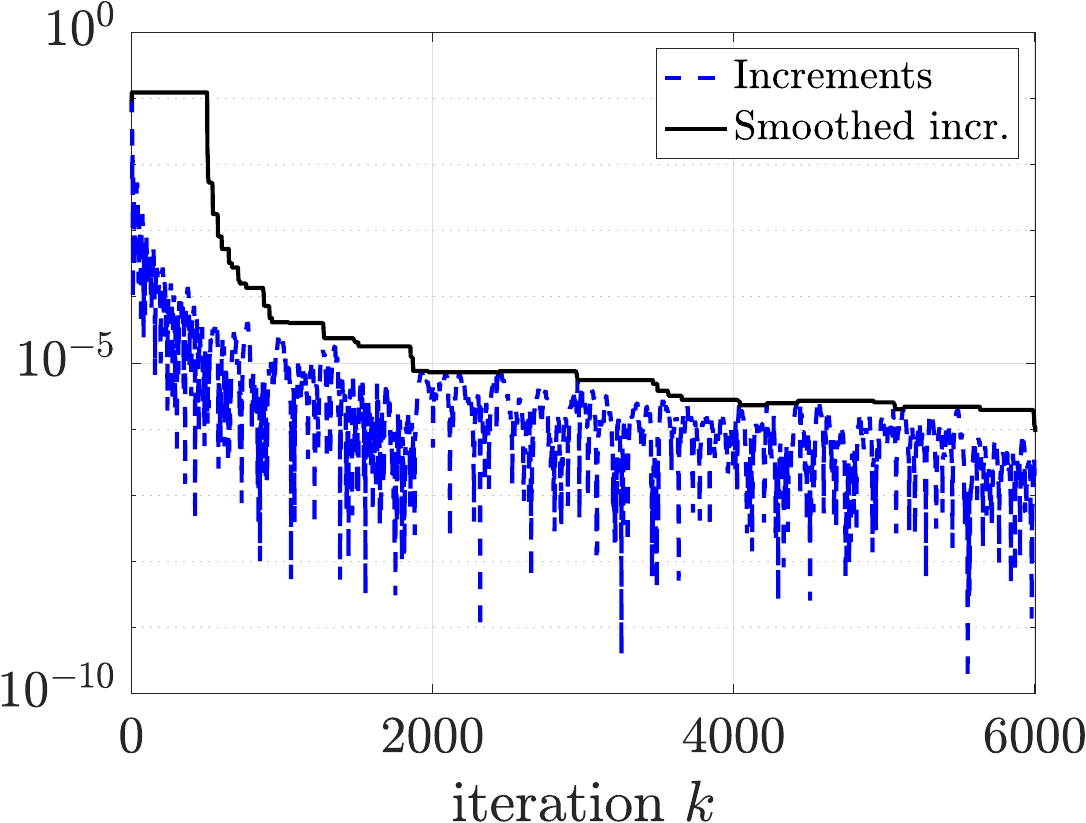}}
\caption{\label{fig:conv_crit}
\textbf{Stopping criterion.} (Normalized) Objective function increments and non linearly smoothed increments as functions of the number of iterations in Algorithm~\ref{alg:CP} (computed from new infections in France).}
\end{figure}

\section{Covid-19 data}
\label{sec:data}

To show the relevance and efficiency of the proposed procedure for the estimation of the evolution of the reproduction number of the pandemic, two real Covid-19 datasets made available by national health authorities are used, different in nature and from different international repositories. 

\subsubsection{Dataset1: Country population level daily new infection counts} \emph{Johns Hopkins University}\footnote{
\href{https://coronavirus.jhu.edu/}{https://coronavirus.jhu.edu/}}
provides access to the cumulated counts of daily new infections (together with deceased and recovered persons), at the entire population level, on a per country basis, for 200+ countries worldwide, since January 1st, 2020, hence, essentially since the earliest stages of the Covid-19 pandemic\footnote{\href{https://raw.githubusercontent.com/CSSEGISandData/COVID-19/master/csse_covid_19_data/csse_covid_19_data/csse_covid_19_time_series/}
{https://raw.githubusercontent.com/CSSEGISandData/COVID-19/master/csse$\_$covid$\_$19$\_$time$\_$series/}
}.  
In the present work, use will be made of the daily new infection counts only, $ \{ Z_{1:T}\}$, from Feb., 15th, 2020 until until August, 31st, 2021, hence $T = 586$ days of pandemic.

\subsubsection{Dataset2: French Counties hospital Level daily new infection counts} \emph{Sant\'e-Publique-France}\footnote{
\href{https://www.santepubliquefrance.fr/}{https://www.santepubliquefrance.fr/}}, the French national health authorities, provides several different datasets related to the Covid pandemic in France.
In the present work, use will be made of the daily hospital-recorded new infection counts, across France on a per \emph{d\'epartement}-basis\footnote{\href{https://www.data.gouv.fr/fr/datasets/}{https://www.data.gouv.fr/fr/datasets/}},
\emph{d\'epartements} (or counties) consisting of the usual French administrative granularity.
These new infections counts are stacked into a ${Z} \in \RR^{T \times D}$ data matrix, for the $101$ French  \emph{d\'epartements}, and estimation of the reproduction number will be conducted jointly for all \emph{d\'epartements}. 
Such data are however available only after March 19th, 2020 and analyzed until August, 31st, 2021, hence $T = 531$ days of pandemic.

\section{Robust estimation from Covid-19 data} 
\label{sec:results}

\subsection{Hyperparameter tuning} 
\label{sec:hyperparam}

One of the issue at hand for a practical use of Eq.~\ref{eq:penal_KL} to estimate jointly $R$ and $O$ lies in the selection of the hyperparameters, $\lambda_{\mathrm{T}}, \lambda_{\mathrm{S}}, $ and $\lambda_{\mathrm{O}}$ that balance the contributions of the regularizations terms one against the others as well as against the data-model fidelity term. 
Such choices have significant impacts on the achieved estimates.
Accurate parameter tuning is thus critical yet may prove time consuming notably when this selection needs to be made for each of the 200+ countries studied here, with possibly very different population sizes or pandemic intensities. 

Automated data driven selection of hyperparameters in functional minimization has been addressed theoretically in several different settings (cf. e.g. \cite{pascal2021automated})
and references therein), often relying on a Stein Unbiased Risk Estimator (SURE). 
This is working well only for large size data, which is however not the case for the Covid-19 pandemic. 

Instead, in the present work, hyperparameter selection was driven by the following dimensional analysis. 
Inspecting Eq.~\ref{eq:penal_KL} shows that the functional to minimize is covariant under the change $( Z, O ;  \lambda_{\mathrm{T}},\lambda_{\mathrm{S}},\lambda_{\mathrm{O}})$ $ \rightarrow $ $( \alpha Z, \alpha O ;  \alpha\lambda_{\mathrm{T}},\alpha\lambda_{\mathrm{S}},\lambda_{\mathrm{O}})$ for any $\alpha > 0$, thus leading to identical minima. 
This led us to propose to use one same set of hyperparameters $(\lambda_{\mathrm{T}},\lambda_{\mathrm{S}},\lambda_{\mathrm{O}})$ valid for all countries, irrespective of their  population sizes or pandemic intensities, by applying minimizations in~\ref{eq:penal_KL} to $ \alpha Z$ instead of $Z$, where $\alpha$ are arbitrary multiplicative constants, that may depend on each country.
We found empirically that setting systematically $ \alpha $ to the empirical standard deviation of the observed $Z$ constitutes an efficient way to account both for the population size and for the severity of the pandemic. 
Precisely, $ \lambda_{\mathrm{T}} = 3.5$ is chosen as a universal (for all countries) time constant, after inspections of numerous different countries over the whole pandemic period, as a valid trade-off between too small --- that would lead to too much variability in the estimation of $R$ and would prevent the use of $R$ to actually assess the strength of the pandemic --- and too large --- that would lead to too few variability hence preventing the possibility of detecting changes in the pandemic dynamics that may follow some sanitary policy decisions (lockdown, vaccination). 
Further, $ \lambda_{\mathrm{O}}  $ is set to $ 0.025 $ to enable the outlier term in the functional to be versatile enough to account for pseudo-seasonal effects modeled as sparse irrelevant measures. 
Finally, for different $D=96$ territories, the connectivity of the graph for metropolitan France contains $475$ edges, indicating that the time and space regularizations will have comparable impacts when $96  \lambda_{time }  \simeq 2\times 475  \lambda_{space } $, which frames the selection of  $ \lambda_{space } $. 
Inspection of the results lead to choose  $ \lambda_{space } = 0.002 $, thus corresponding to a mild spatial regularization.

\begin{figure}[htbp]
\centerline{
\includegraphics[width=\linewidth]{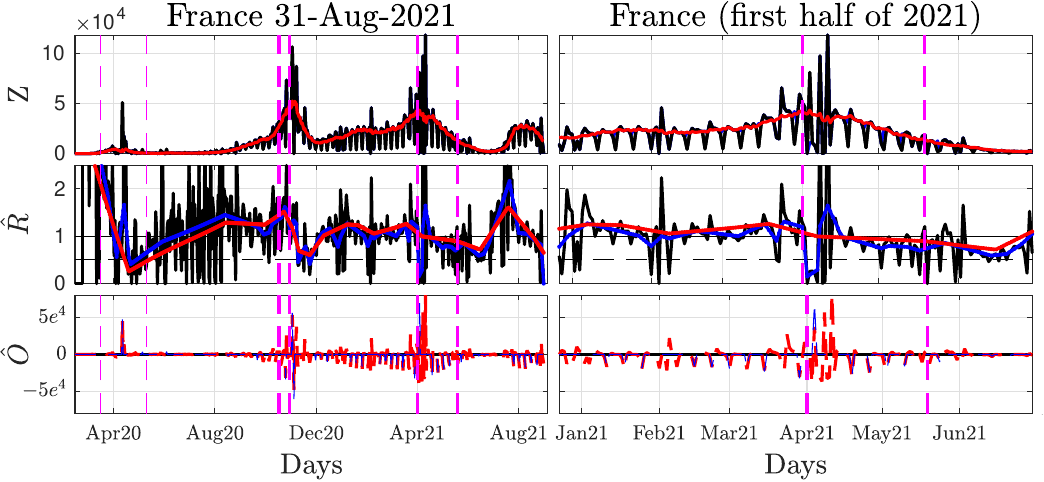}
}
\caption{{\bf France (six first months in 2021).} 
Top: $Z$ and $Z-\hat O_{O}$, 
(black and red)~;
Middle: $\hat R^{MLE}, \hat R_{O}, \hat R_{noO}$ (black, red and blue)~; 
Bottom: $\hat O_{O}, \hat O_{noO}$ (red and blue).}
\label{fig:france}
\end{figure}

\subsection{Temporal evolution for each country independently} 
\label{sec:country}

Because Public Health Policies were not or weakly coordinated across countries, even for closely intertwined spaces such as Western European countries, 
population level daily new infection counts (Dataset1, JHU) are first analyzed independently for each country, that is, without spatial regularization, or setting $\lambda_{\mathrm{S}} \equiv 0$ in \eqref{eq:penal_KL}.

Estimates of the reproduction numbers and the outliers, $\hat R_{O}, \hat O_{O}$ are thus obtained for each country by applying \eqref{eq:penal_KL} to 
daily infection counts $Z_{1:T}$, after normalisation  by standard deviations computed across the whole pandemic period, using the same hyperparameters for all countries : $(\lambda_{\mathrm{T}},\lambda_{\mathrm{S}},\lambda_{\mathrm{O}})=(3.5,0,0.025)$. 

To assess the relevance of the proposed one-step procedure, $\hat R_{O}, \hat O_{O}$ are compared against estimates obtained from a two-step procedure developed in~\cite{abry2020spatial},  $\hat R_{noO}$:
First, a  sliding-median over a 7-day window is applied to $Z_{1:T}$, and values that depart from the window median by  $\pm 2.5$ local (in-window) standard deviation are replaced by the window median, yielding estimates of outliers referred to as $\hat O_{noO}$~; 
Second, an estimate $\hat R_{noO}$ of $R$ is obtained by applying Eq.~\ref{eq:penal_KL} with $(\lambda_{\mathrm{T}},\lambda_{\mathrm{S}},\lambda_{\mathrm{O}})=(3.5, 0, +\infty)$, to the a priori denoised infection counts after normalization by standard deviation.

Fig.~\ref{fig:france} compares the estimates $ \hat R^{MLE}$, $(\hat R_{O}, \hat O_{O})$  and $(\hat R_{noO},\hat O_{noO}) $, for France across the whole pandemic period, and also focuses on a recent narrow period of time to ease illustration and analysis. 
Fig.~\ref{fig:france} shows that, for the pseudo-seasonal effect related to new infection count under-reporting due to non-working days, both procedures essentially correctly detect and estimate the corresponding outliers. 
For the subsequent working-day over-reporting performed by the French public health authority,  Fig.~\ref{fig:france} clearly shows that
the proposed one-step procedure correctly accounts for such \emph{small} outliers, thus yielding a smooth and outlier-robust estimation of $R$, while mis-reporting are missed by the two-step procedure, that thus produces estimates with numerous discontinuities irrelevant in terms of pandemic monitoring and clearly driven by these mis-reporting. 
The period of early April 2021 is of great interest, corresponding to a \emph{long week-end effect}, with a normal week-end, followed by an alternation of working and non working days, resulting in anomalously low infection counts, followed by high counts (misreport correction), during intertwined and following working days. 
The two-step procedure yields for that period totally irrelevant estimates, while the proposed one-step procedure remains satisfactorily insensitive to such count misreports. 
These pseudo-seasonal or long week-end effects are not specific to France and can be observed across numerous other countries.

Fig.~\ref{fig:countries} further compares the estimates $ \hat R^{MLE}$, $(\hat R_{O}, \hat O_{O})$  and $(\hat R_{noO},\hat O_{noO}) $ for several different countries.
Fig.~\ref{fig:countries} clearly shows that 
i) $ \hat R^{MLE}$ constitute a very irregular estimate of $R$ along time that cannot be used at all in practice~; 
ii) compared to $(\hat R_{noO}, \hat O_{noO})$,   $(\hat R_{O}, \hat O_{O})$ better accounts for both outliers and pseudo-seasonal effects and yields smooth and regular along time and outlier-robust estimations of $R$.
 $\hat R_{O} $ is thus far more likely to represent the actual, less biased and realistic assessment of the pandemic time evolution in a given country. 
Estimates $(\hat R_{O}, \hat O_{O})$ are updated on a daily basis for 200+ countries across the world, and made publicly available via interactive and animated maps (cf. Section~\ref{sec:conclusions}).

\begin{figure}[tbp]
\centerline{\includegraphics[width=\linewidth]{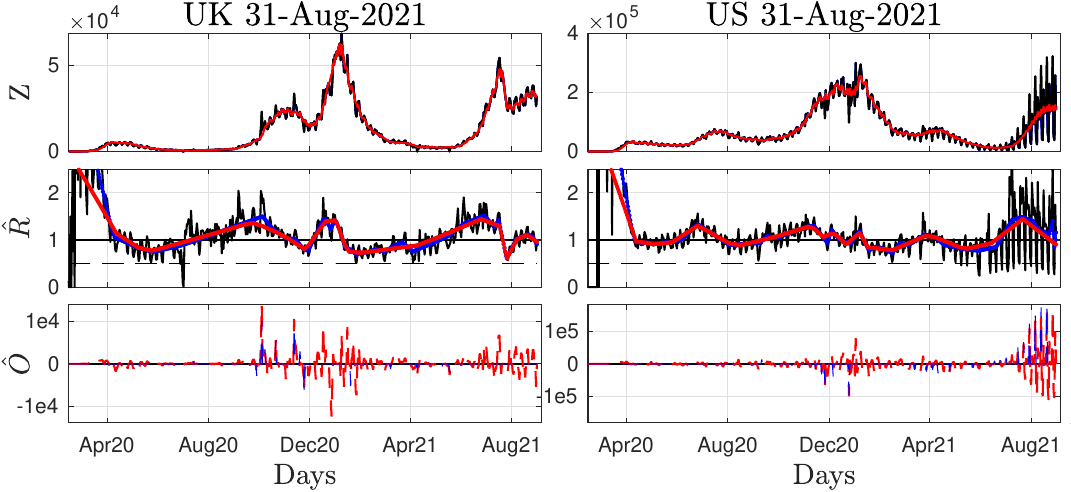}}
\centerline{\includegraphics[width=\linewidth]{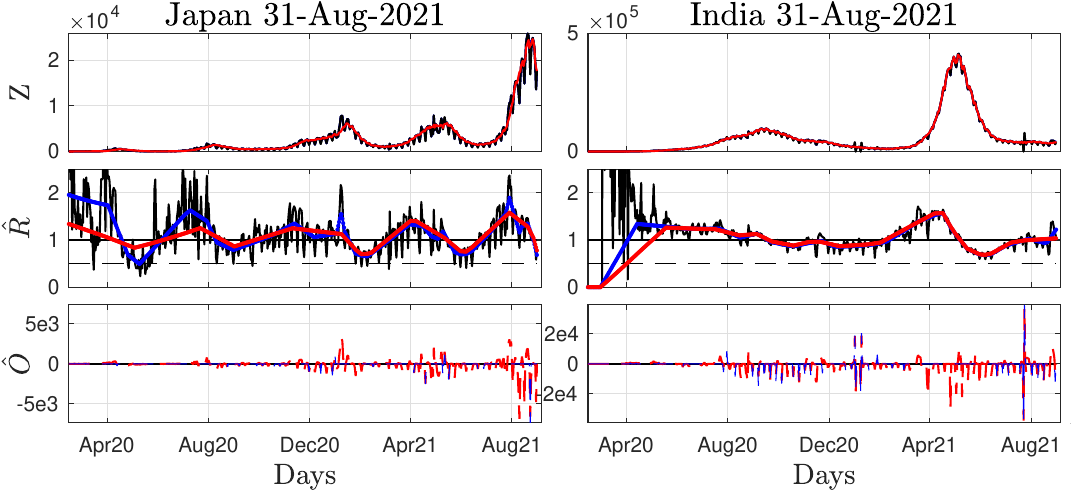}}
\centerline{\includegraphics[width=\linewidth]{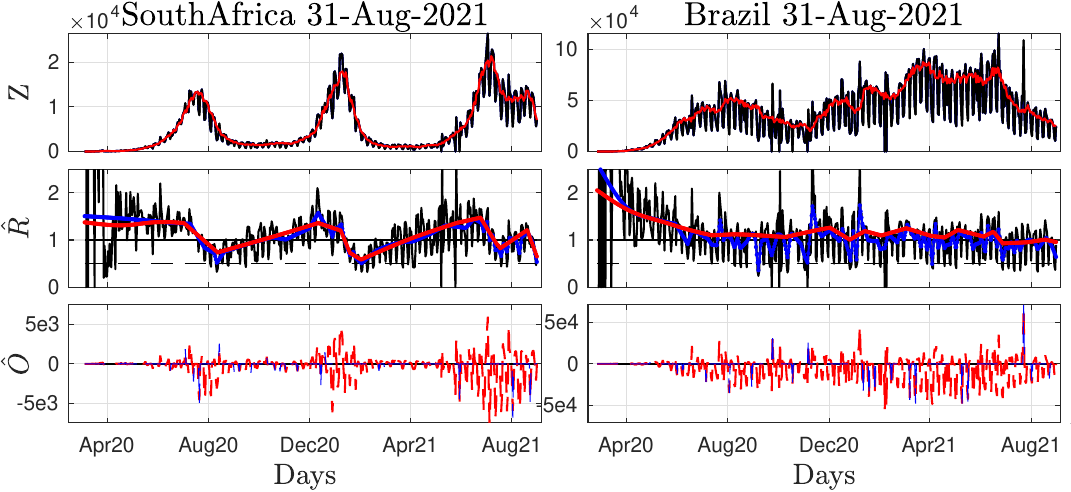}}
\caption{\label{fig:countries} {\bf Different countries.} ,Top: $Z$ and $Z-\hat O_{O}$ (black and red)~;
Middle: $(\hat R^{MLE}, \hat R_{O}, \hat R_{noO})$ (black, red and blue)~; 
Bottom: $(\hat O_{O}, \hat O_{noO})$ (red and blue).
Estimates $(\hat R_{O}, \hat O_{O})$
are updated on a daily basis for 200+ countries across the world, and made publicly available via interactive and animated maps available (cf. Section~\ref{sec:conclusions}).} 
\end{figure}

\subsection{Joint space-time Evolution for metropolitan France} 
\label{sec:france}

Let us now consider multivariate daily infections counts, $\mathbf{Z} \in \RR^{T \times D}$, for the $D=96$ \emph{d\'epartments} of metropolitan France, considered as the vertices of the graph $G$ in Eq.~\ref{eq:penal_KL}, with edges between \emph{d\'epartments}  sharing a terrestrial boundary.
Estimates $(\hat{\mathbf{R}}_{O,S}, \hat{\mathbf{O}}_{O,S}) \in \RR^{T \times D}$ are obtained by applying Eq.~\ref{eq:penal_KL} to ${Z} \in \RR^{T \times D}$, after standardization independently per \emph{d\'epartment}, 
with $( \lambda_{\mathrm{T}},\lambda_{\mathrm{S}},\lambda_{\mathrm{O}}) $ $ = (3.5, 0.002, 0.025)$ (cf. Section~\ref{sec:hyperparam}). 

To illustrate the relevance of  $(\hat{\mathbf{R}}_{O,S}, \hat{\mathbf{O}}_{O,S})$, 
they are compared against estimates obtained without spatial regularization and without accounting for outliers $(\hat{\mathbf{R}}_{noO,noS}, \hat{\mathbf{O}}_{noO,noS})$ ($( \lambda_{\mathrm{T}},\lambda_{\mathrm{S}},\lambda_{\mathrm{O}}) $ $ = (3.5, 0, +\infty)$), 
without spatial regularization but accounting for outliers  $(\hat{\mathbf{R}}_{O,noS}, \hat{\mathbf{O}}_{O,noS})$ ($( \lambda_{\mathrm{T}},\lambda_{\mathrm{S}},\lambda_{\mathrm{O}}) $ $ = (3.5, 0, 0.025)$),  
with spatial regularization but without accounting for outliers  $(\hat{\mathbf{R}}_{noO,S}, \hat{\mathbf{O}}_{noO,S})$ ($( \lambda_{\mathrm{T}},\lambda_{\mathrm{S}},\lambda_{\mathrm{O}}) $ $ = (3.5, 0.002, +\infty))$.

These estimates are first compared as functions of time for two \emph{d\'epartments} in Fig.~\ref{fig:RJ}, which shows that estimates that do not account for outliers are far too irregular for being useful in pandemic monitoring, even with spatial regularization. 
To the converse, $\hat{\mathbf{R}}_{O,S}$ that results from both accounting for outliers and spatial regularization yields very regular estimates likely to reflect a relevant assessment of the pandemic. 
The proposed estimation procedure is of particular interest when applied to territories with granularity level as is the case for the French \emph{d\'epartments}   (with on average slightly less than one million inhabitants) or during low activity phases of the pandemic, such as early summer 2020 (days 95 to 155). 
Spatial regularization is of particular relevance for theses two chosen \emph{d\'epartments}.
Indeed, \emph{d\'epartment Paris (Id=75)} corresponds to the city of Paris, thus a small territory yet with massive connection to the surrounding \emph{d\'epartments} : estimates of $R$ within Paris, can thus not notably differ from those of these so-called \emph{Paris-Crown} \emph{d\'epartments}.
The same holds for \emph{d\'epartment Rh\^one (Id=69)}, of very small surface (for historical reasons) yet acting, with the city of Lyon as a massive transit hub, for a number of surrounding \emph{d\'epartments}.

\begin{figure}[tbp]
    \centering
\centerline{\includegraphics[width=\linewidth]{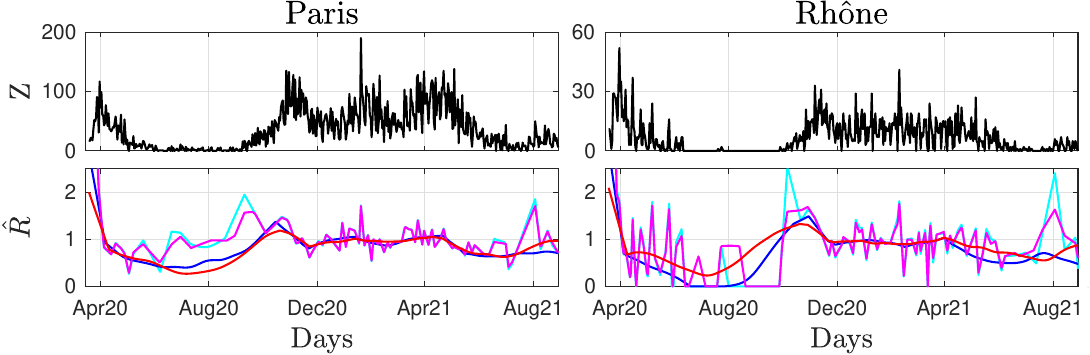}}
        \caption{\label{fig:RJ}{\bf Estimates of $R_t$ for \emph{d\'epartment Paris (Id=75)} (right) and \emph{d\'epartment Rh\^one (Id=69)} (left):}
        $\hat{\mathbf{R}}_{O,S} $ (accounting for outliers and with spatial regularization, red), $\hat{\mathbf{R}}_{noO,S} $ (not accounting for outliers but with spatial regularization, magenta), $\hat{\mathbf{R}}_{O,noS} $ (accounting for outliers but without spatial regularization, blue), $\hat{\mathbf{R}}_{noO,noS} $ (not accounting for outliers and without spatial regularization, cyan).}
\end{figure}

Further, the estimates of $R$ are for continental France are reported in Fig.~\ref{fig:RJFrance}) for three different days.
Estimates of  $\hat{\mathbf{R}}_{O,S}$ with spatial regularization and accounting for outliers (right most plots) permits a far clearer assessment of the status of the pandemic across metropolitan France.
Estimates $\hat{\mathbf{R}}_{noO,S}$ and $\hat{\mathbf{R}}_{noO,noS}$ show far too much spatial variability to be realistic estimates.
Estimates $\hat{\mathbf{R}}_{noO,S}$ improve spatial regularity, but still show significant variability induced by data low quality. 
Estimates $\hat{\mathbf{R}}_{O,S}$, with spatial regularization and accounting for outliers (right most plots), likely yield the most likely and relevant estimates each day and permit a far clearer assessment of the status of the pandemic across the connected metropolitan France \emph{d\'epartements}.
March, 31st,  corresponds to the pandemic 1st-wave maximum and $\hat{\mathbf{R}}_{O,S}$  shows a clear North-East/South-West gradient in the strength of the pandemic~; 
March, 31st, 2021 corresponds to the start of the third lockdown period, referred to as \emph{couvre-feu} (corfew) in France and $\hat{\mathbf{R}}_{O,S}$  shows that the pandemic was significantly progressing uniformly across all France~; 
August, 31st, 2021 corresponds to the time of submission of this work and shows that the pandemic is globally decreasing in France, yet non uniformly and still active in large areas.
Estimates $\hat{\mathbf{R}}_{O,S}$ are updated automatically on a daily basis for the 96 continental France  \emph{d\'epartements} and made publicly available via interactive and animated maps (cf. Section~\ref{sec:conclusions}).
 
 \begin{figure}[tbp]
    \centering
    \centerline{\includegraphics[width=1.1\linewidth]{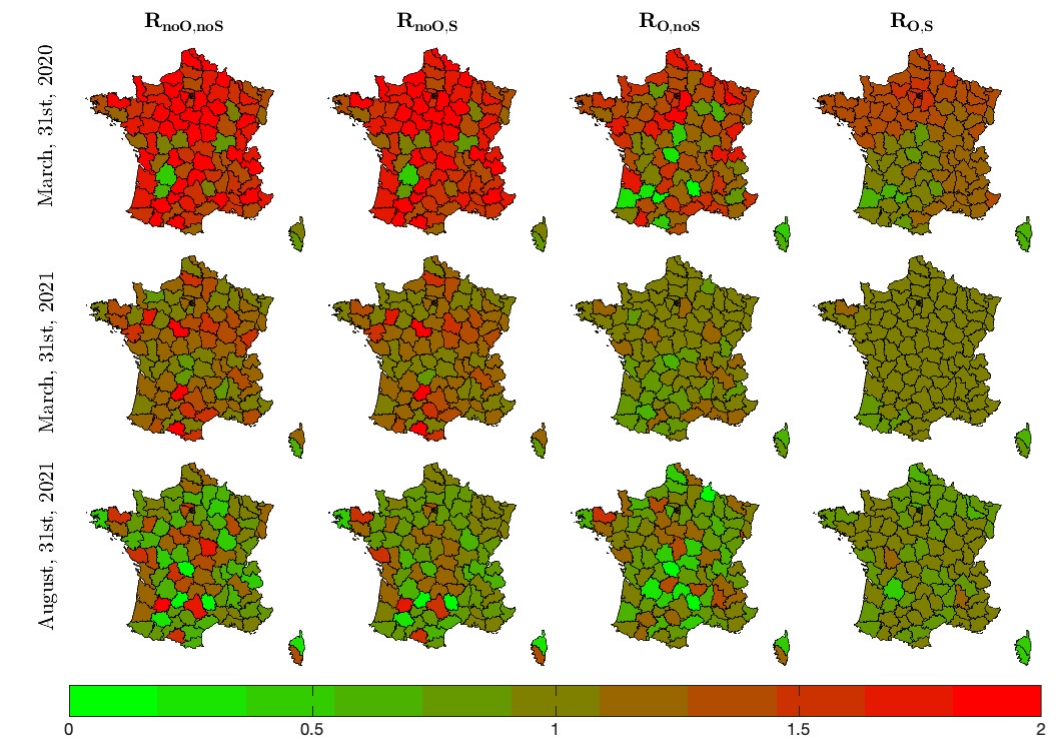}}
    \caption{    \label{fig:RJFrance}{\bf Compared estimates of $R_t^{(d)}$ across metropolitan France}  for three different days along the pandemic.} 
\end{figure}

\section{Conclusion}
\label{sec:conclusions}
The present work developed and assessed an \emph{inverse problem}-type procedure to estimate the spatio-temporal evolution of the Covid-19 reproduction number, $R$, robust to the low quality of the pandemic data.
The devised procedure relies first on a detailed analysis of the data quality that suggested that most count misreports, be they mild and seasonal (sundays, week-end, days-off) or large and random (report failure), can be efficiently accounted for as sparse outliers. 
Second, it is based on devising carefully a functional that balances a data-model fidelity term (Poisson distribution and outliers), and regularity in time and space properties that $\hat R$ needs to fulfill to be of actual relevance in practical pandemic monitoring. 
Fast and efficient algorithms were devised to minimize this functional and their convergence was studied. 
Applied to real Covid-19, the procedure was shown to yield meaningful estimates of $R$, that thus provides a relevant assessment of the spatio-temporal evolution of the pandemic that can be involved as part of a decision strategy for designing Covid-19 counter measures. 
Indeed, at a current time, the procedure outputs estimates of $R$ for the entire pandemic period, thus permitting to assess a posteriori  if and how the implementation of a  sanitary policy decision (lockdown, curfew,\ldots) impacted the evolution of the pandemic. 
Also and importantly, while the proposed procedure is not forecasting $R$, it is actually implementing a \emph{nowcasting} approach: 
Not only estimates of $R$ today are provided, but the imposed piecewise smoothness constraints provide epidemiologist with a local trend around at current time indicating whether the pandemic is progressing or regressing.  
Spatial regularization also permits to robustly assess homogeneity or heterogeneity of the pandemic across related territories, thus permitting to decide on local or global (nation-wide) measures. 

Estimates are updated automatically on a daily basis for 200+ countries worldwide and for several related territories (Counties in France, States in the USA).
They are made publicly available via downloadable text files or interactive and animated maps\footnote{Animated and interactive maps available at \url{perso.ens-lyon.fr/patrice.abry/}, \url{http://barthes.enssib.fr/coronavirus/cartes/Rmonde/} and \url{http://barthes.enssib.fr/coronavirus/cartes/RFrance/} are conceived and deployed in collaboration with Eric Guichard (ENSSIB, Triangle ENS de Lyon, France), who is here gratefully acknowledged.}.
Procedures will be made publicly available. 

Future investigations to improve estimation and monitoring include the construction of confidence intervals for the estimates and automated data-driven selection of the hyperparameters, using adaptive strategy that adjust the variations and phases of the pandemic. 

We believe that the present work constitutes a significant step forward toward the practical use of the proposed procedure permits an  actual \emph{real-time} and \emph{on-the-fly} monitoring of the Covid-19 pandemics.
These tools are ready to be put at work immediately at the outbreak of potential future pandemics. 
Finally, the automated production of these daily up-dates is also intended as a scientific tool to favor transdisciplinary scientific work against Covid-19 impact on the society at large.

 \bibliographystyle{plain}
\bibliography{CovidBiblio.bib}

\end{document}